\documentclass[english, 11pt, letterpaper]{scrartcl}

\title{Simple Approximation Algorithms\\ for Minimizing the Total Weighted Completion Time of Precedence-Constrained Jobs}

\author{Sven~Jäger\thanks{\href{mailto:sven.jaeger@math.rptu.de}{sven.jaeger@math.rptu.de}, RPTU Kaiserslautern-Landau, Paul-Ehrlich-Straße 14, 67663 Kaiserslautern} \and Philipp~Warode\thanks{\href{mailto:philipp.warode@hu-berlin.de}{philipp.warode@hu-berlin.de}, Humboldt-Universität zu Berlin, Unter den Linden 6, 10099 Berlin}}

\usepackage[utf8]{inputenc} 
\usepackage[textwidth=470pt, textheight=610pt]{geometry}

\usepackage{amsmath}
\usepackage{amssymb}
\usepackage{amsthm}
\usepackage{mathtools,nicefrac}

\usepackage{tikz}
\usetikzlibrary{calc}

\definecolor{Red}{HTML}{DB001F}
\definecolor{Green}{HTML}{238224}
\definecolor{Blue}{HTML}{3840FE}
\definecolor{Purple}{HTML}{730087}
\newcommand{\purplename}{purple}
\colorlet{firstcolor}{Red}

\definecolor{secondcolor}{HTML}{F68C21}

\definecolor{thirdcolor}{HTML}{FFED00}

\colorlet{fourthcolor}{Green}

\colorlet{fifthcolor}{Blue}

\colorlet{sixthcolor}{Purple}

\colorlet{availablecolor}{secondcolor}
\let\availablecolorname\secondcolorname
\tikzstyle{job}=[draw=black, fill opacity=.5, text=black, text opacity=1]

\usepackage{placeins}
\usepackage{caption}
\usepackage[colorlinks,breaklinks=true,urlcolor=black,linkcolor=Red,citecolor=Green]{hyperref}
\usepackage[noabbrev,capitalise]{cleveref}
\crefname{equation}{equation}{equations}

\usepackage[citetracker, backend=biber, style=mynumeric, related=true, hyperref=true, eprint=true, isbn=false, date=year, sorting=nyt, maxcitenames=3, maxbibnames=12, giveninits=true, url=false, doi=false]{biblatex}

\usepackage{algorithm}
\usepackage{algpseudocodex}

\usepackage{booktabs,array}

\newtheorem{theorem}{Theorem}[section]
\newtheorem{corollary}[theorem]{Corollary}
\newtheorem{lemma}[theorem]{Lemma}

\theoremstyle{remark}
\newtheorem{example}{Example}

\definecolor{colorsven}{RGB}{119,182,186}

\definecolor{colorphilipp}{RGB}{200,66,66}

\newcommand{\opt}{\ensuremath{\mathrm{OPT}}}
\newcommand{\alg}{\ensuremath{\mathrm{ALG}}}
\DeclareMathOperator*{\argmin}{arg\,min}

\let\oldparagraph\paragraph
\renewcommand{\paragraph}[1]{\oldparagraph{#1.}}

\graphicspath{{figures/}}
\bibliography{references}

\begin{document}

\maketitle

\begin{abstract}
 We consider the precedence-constrained scheduling problem to minimize the total weighted completion time. For a single machine several $2$-approximation algorithms are known, which are based on linear programming and network flows. We show that the same ratio is achieved by a simple weighted round-robin rule. Moreover, for preemptive scheduling on identical parallel machines, we give a strongly polynomial $3$-approximation, which computes processing rates by solving a sequence of parametric flow problems. This matches the best known constant performance guarantee, previously attained only by a weakly polynomial LP-based algorithm. Our algorithms are both also applicable in non-clairvoyant scheduling, where processing times are initially unknown. In this setting, our performance guarantees improve upon the best competitive ratio of $8$ known so far.
\end{abstract}

\section{Introduction}

\FloatBarrier

Scheduling jobs with precedence constraints so as to minimize the sum of weighted completion times is a widely studied problem in the field of approximation and online algorithms. We consider this problem on a single machine, as well as on identical parallel machines when preemption is allowed. These problems, denoted in the $3$-field notation for scheduling problems~\cite{GLLK79} as $1\mid\text{prec}\mid\sum w_j C_j$ and $\mathrm P\mid \mathrm{prec}, \mathrm{pmtn} \mid \sum w_j C_j$, respectively, are both strongly NP-hard~\cite{Law78}. We present simpler and faster approximation algorithms for both problems, achieving the same performance guarantees as the best previously known algorithms.

\Textcite{Pis92} observed in 1992 that the single-machine scheduling problem is a special case of the so-called submodular ordering problem, for which he gave a $2$-approximation algorithm and thus the first $2$-approximation for the scheduling problem. Since then, a whole set of further $2$-approximation algorithms have been developed for this problem, and under a variant of the Unique Games Conjecture no better guarantee is possible~\cite{BK09}. The algorithm of \textcite{HSSW97} is based on a linear programming relaxation with an exponential number of efficiently separable constraints~\cite{Que93} and hence relies on the ellipsoid method. 
All approximation algorithms developed in the sequel are purely combinatorial and based on network flows. \Textcite{CH99} consider a special LP relaxation with only two variables in each constraint, which can be solved to optimality by a minimum capacity cut computation~\cite{HMNT93}, and then apply list scheduling to the resulting LP completion times. \Textcite{CM99,MQW03,Pis03} all determine a Sidney decomposition~\cite{Sid75} and order the jobs arbitrarily within each block of the decomposition. The flow-based methods for computing the Sidney decomposition differ slightly in the three algorithms. While \Textcite{CM99,Pis03} solve a separate (parametric) maximum flow problem for each Sidney block, \citeauthor{MQW03} manage to compute the entire decomposition by a single application of the algorithm by \textcite{GGT89} for computing all breakpoints of a parametric maximum flow problem. This yields the to date fastest $2$-approximation algorithm for the problem, running in time~$O(n^3)$ for $n$ jobs. As to simplicity, the algorithm of \textcite{Pis03}, relying only on \citeauthor{Sid75}'s result and a standard maximum flow computation, may be considered the easiest. For identical parallel machines, \textcite{HSSW97} gave a $3$-approximation algorithm, which is still the best known constant performance guarantee for this problem. Their algorithm is based on the solution of a linear programming relaxation and thus runs only in weakly polynomial time.

We show that a very simple weighted round-robin rule, which always alternates between all available jobs, also achieves the performance guarantee~$2$ on a single machine. Here we call a job available if it is unfinished but all its predecessors have been completed. The algorithm passes the weight of each job waiting for an uncompleted predecessor to an arbitrary available predecessor. It then runs each available job at a rate proportional to its collected weight, thus constructing a schedule with infinitesimally small processing intervals. Note that this algorithm is non-clairvoyant, i.e., it needs no a priori knowledge of the processing times. When the processing times are known upfront, it can easily be transformed to a non-preemptive algorithm by scheduling the jobs in the order of completion times in the computed preemptive schedule. The entire algorithm runs in time~$O(n^2)$.

For identical parallel machines we present a $3$-competitive non-clairvoyant algorithm, which simplifies an algorithm by \textcite{GGKS19} and is based on a parametric maximum flow computation. It generalizes our algorithm for a single machine in the sense that, when applied to a single machine, it will return the same result. However, to accommodate the greater generality, it is more complicated than the single-machine algorithm, resulting in a longer running time. Similarly as above, the resulting schedule can be transformed into a schedule with only a finite number of preemptions. This gives the first strongly polynomial $3$-approximation algorithm for this problem.

Our contribution is thus twofold: First, we obtain simpler and faster approximation algorithms for clairvoyant precedence-constrained scheduling, see \cref{tab1}.
\begin{table}
 \caption{Running times and main ingredients of old and new approximation algorithms for precedence-constrained scheduling}
 \label{tab1}
 \centering

 \begin{tabular}{ccc}
  \toprule
  & \multicolumn{1}{c}{old} & new \\\midrule
  $2$-approximation for & $O(n^3)$  \cite{MQW03} & $O(n^2)$  \\
  $1\mid \text{prec} \mid \sum w_j C_j$ & (parametric flows) & (weighted round-robin) \\\midrule
  $3$-approximation for & weakly polynomial \cite{HSSW97} & $O(n^4)$ \\
  $\mathrm P \mid \text{prec}, \text{pmtn} \mid \sum w_j C_j$ & (LP-based) & (parametric flows) \\\bottomrule
 \end{tabular}
\end{table}%
Second, we get improved competitive ratios for non-clairvoyant scheduling with precedence constraints. Our non-clairvoyant algorithms for a single machine and for identical parallel machines generalize respective algorithms by \textcite{LLMS23} for out-forest precedence constraints, which were shown to be $4$-competitive and $6$-competitive, respectively. For general precedence constraints, the algorithm by \textcite{GGKS19} was improved by \textcite{Jaeg21} to give $8$-competitiveness, which has previously been the best known bound for this problem. In the absence of precedence constraints, $2$-competitiveness was proved by \textcite{KC03,BBEM12}, and no non-clairvoyant algorithm can achieve a better competitive ratio~\cite{MPT94}. \Cref{tab2} provides an overview of the old and new performance guarantees for non-clairvoyant scheduling.
\begin{table}
 \caption{Old and new competitiveness results for non-clairvoyant scheduling}
 \label{tab2}
 \centering
 \begin{tabular}{clc}
  \toprule
  & \multicolumn{1}{c}{old} & new \\\midrule
  $1\mid \text{pmtn} \mid \sum w_j C_j$ & $2$ \cite{KC03} \\
  $1 \mid \text{out-forest}, \text{pmtn} \mid \sum w_j C_j$ & $4$ \cite{LLMS23} & $2$ \\
  $1 \mid \text{prec}, \text{pmtn} \mid \sum w_j C_j$ & $8$ \cite{Jaeg21} & $2$ \\\midrule
  $\mathrm P \mid \text{pmtn} \mid \sum w_j C_j$ & $2$ \cite{BBEM12} \\
  $\mathrm P \mid \text{out-forest}, \text{pmtn} \mid \sum w_j C_j$ & $6$ \cite{LLMS23} & $3$ \\
  $\mathrm P \mid \text{prec}, \text{pmtn} \mid \sum w_j C_j$ & $8$ \cite{Jaeg21} & $3$ \\\bottomrule
 \end{tabular}
\end{table}

Not only the presented algorithms are simple, but also their analysis. The performance guarantee is derived in each case by means of an induction over the number of jobs, inspired by the analysis of \textcite{BBEM12}. For a single machine, it is fully self-contained, using only elementary calculations and not relying on any theoretical foundations like linear programming theory or network flows, so it could be taught in an introductory algorithms course. For identical parallel machines, only network flow theory is needed. This is in contrast to previous analyses for non-clairvoyant scheduling with precedence constraints~\cite{GGKS19,LLMS23}, which were based on dual fitting for an appropriate LP relaxation.

\paragraph{Furhter related results}

For some special classes of precedence graphs the single-machine problem can be solved in polynomial time. A classical result states that this is the case for series-parallel graphs~\cite{Knu,Law78}. This was generalized by \textcite{AM09} to $2$-dimensional precedence orders, by proving that the scheduling problem is a special case of the vertex cover problem. A survey on the complexity of scheduling problems with special precedence constraints is given by \textcite{PB18}. A generalization of the result of \citeauthor{AM09} yields approximation factors below $2$ for several further precedence graph structures: \textcite{AMMS11} showed that if the precedence order is $k$-dimensional (and the $k$ realizing linear orders are given), then there is a $(2-2/k)$-approximation algorithm. This implies a $4/3$-approximation algorithm for convex bipartite orders and for semiorders. Moreover, by generalizing to fractional dimensions, they proved for any $\Delta \ge 1$ that if no job has more than $\Delta$ predecessors or no job has more than $\Delta$ successors, then there is a $2/(1+1/\Delta)$-approximation. \Textcite{Sit21} recently gave a polynomial-time approximation scheme for interval orders.

For integral processing times, preemptive scheduling with precedence constraints is closely related to non-preemptive scheduling of precedence-constrained jobs with unit processing times. On the one hand, given an instance with arbitrary processing times, every job can be replaced by a chain of $p_j$ unit jobs such that the entire weight lies on the last job, see e.g.~\cite{BK99}. In the resulting unit processing time instance, preemptions are not useful. Hence, any approximate solution for non-preemptive scheduling of the new instance corresponds to a preemptive schedule for the original instance with the same performance guarantee. This reduction is, however, only pseudopolynomial. On the other hand, \citeauthor*{HSSW97}'s algorithm computes a schedule where preemptions occur only at integer times. Hence, when applied to jobs with unit processing times, it will not introduce any preemptions. Therefore, it is a $3$-approximation algorithm also for $\mathrm P \mid \text{prec}, p_j=1\mid \sum w_j C_j$. For this problem a better $(1+\sqrt 2)$-approximation was recently developed by \textcite{Li20}, which implies a pseudopolynomial approximation with this guarantee for preemptive scheduling of general jobs.

For non-preemptive scheduling of precedence-constrained jobs with general processing times on identical parallel machines the currently best known performance guarantee is $2+2\ln 2 + \varepsilon \approx 3.386 + \varepsilon$ for any $\varepsilon > 0$, achieved by another algorithm of \textcite{Li20}.

Makespan minimization is a special case of our problem because this objective function can be modeled by adding one dummy job that has to wait for every other job. For this case Graham's list scheduling algorithm has performance guarantee~$2$~\cite{Gra66}, and the $(2-\varepsilon)$-hardness under the Unique Games Conjecture variant still applies~\cite{Sve11}.

The algorithms presented in this paper only work when all jobs are released at the beginning and can only be blocked due to unfinished predecessors. In contrast, the $3$-approximation algorithm of \textcite{HSSW97} for identical parallel machines can still be applied when each job has an individual release date, whereas for a single machine no known approximation algorithm with performance guarantee exactly $2$ can handle release dates. However, a $(2+\varepsilon)$-approximation was provided by \textcite{SY18}. For non-clairvoyant scheduling with release dates the algorithm of \textcite{Jaeg21} admits the currently best known performance guarantee of $8$.

Another generalization of the problem studied in this paper concerns the objective function: \Textcite{SV16} gave a universal $2$-approximation algorithm for $1 \mid \text{prec} \mid \sum w_j f(C_j)$ for all concave functions~$f$. Moreover, they showed that for any given concave functions~$f_j$ for all jobs~$j$, the problem $1\mid \text{prec} \mid \sum f_j(C_j)$ admits a $(2+\varepsilon)$-approximation.

In our non-clairvoyant algorithms, we always assume that the weights and the precedence graph are known precisely. When the precedence constraints form an out-forest, the knowledge of the entire graph is not necessary, but it would be enough to know only the total weight of jobs depending on each available job. \Textcite{LLMS23} study the case when these aggregated weights are also not known exactly but only a prediction on their value is at hand. Using time sharing, they derive an algorithm whose performance guarantee depends on the maximum distortion factor between the predicted weights and the real weights, but is capped at the width of the precedence order.

\paragraph{Structure of the paper}

Some basics and notation are set up in \cref{preliminaries}. In \cref{non-clairvoyant single,non-clairvoyant parallel} the non-clairvoyant algorithms for a single and for identical parallel machines, respectively, are presented. It is discussed in \cref{approximations} how these algorithms can be converted to approximation algorithms for the problems $1\mid\text{prec}\mid\sum w_j C_j$ and $\mathrm P\mid\text{prec}, \text{pmtn}\mid\sum w_j C_j$.

\section{Preliminaries and Notation} \label{preliminaries}

Assume that we are given a set of jobs~$N = \{1,\dotsc,n\}$ with processing times~$p_j \ge 0$, weights~$w_j \ge 0$, and precedence constraints~$A \subset N \times N$ such that $(N, A)$ is a directed acyclic graph. A job~$k$ is \emph{available} if it is not yet completed but all jobs~$k$ with $(k,j) \in A$ are.
A schedule~$\mathrm S$ assigns to each available job a processing rate~$R_j^{\mathrm S}(t) \in [0,1]$ at any time~$t \ge 0$ so that the sum of all processing rates never exceeds~$1$. The processing time of a job~$j$ elapsed before time~$t$ is $Y_j^{\mathrm S}(t) = \int_0^t R_j^{\mathrm S}(s) \,\mathrm d s$. A job~$j$ is completed at the time~$C_j^{\mathrm S} = \min \{t \ge 0 \mid Y_j(t) \ge p_j\}$, at which its elapsed time reaches its required processing time. The goal is to minimize the total weighted completion time~$\sum_{j=1}^n w_j C_j^{\mathrm S}$ of all jobs. We omit the schedule~$\mathrm S$ in these notations if it is clear from context or when the schedule is in the process of being constructed. The performance of an algorithm is assessed by comparing the achieved objective value to the objective value of an optimal schedule with full information in the worst case.

For $j \in N$ let $S(j)$ be the set containing $j$ and all its successors in the precedence graph, i.e., all nodes reachable from node~$j$. Moreover, for $t \ge 0$ let $U_t \coloneqq \{j \in N \mid C_j > t\}$ be the set of unfinished jobs, and let $F_t \subseteq U_t$ be the set of available jobs at time~$t$. For a subset $J \subseteq N$ of jobs we will write $w(J) \coloneqq \sum_{j \in J} w_j$.

In non-clairvoyant scheduling the weights and precedence constraints of all jobs are known at the beginning, but the processing times are unknown and are only revealed at the moment when the jobs are completed.

\section{Non-clairvoyant Scheduling on a Single Machine} \label{non-clairvoyant single}

In non-clairvoyant scheduling one specifies a tentative schedule, which can be updated whenever a job is completed. The tentative schedules occurring in our non-clairvoyant algorithm always process every job at a constant rate. These rates are computed by \cref{alg1},
\begin{algorithm}[h]
 \caption{Processing rates at time~$t$ on a single machine}
 \label{alg1}
 \begin{algorithmic}
  \State Let $U' \gets U_t$.
  \LComment{Compute depth-first search forest}
  \ForAll{$i \in F_t$}
   \State set $T(i) \gets S(i) \cap U'$;
   \State set $U' \gets U' \setminus T(i)$.
  \EndFor
  \LComment{Set processing rates}
  \State Process each job~$i \in F_t$ at rate $R_i(t) \gets \frac{\sum_{j \in T(i)} w_k}{\sum_{j \in U_t} w_j}$.
 \end{algorithmic}
\end{algorithm}
which iterates over all available jobs and assigns to every job all of its yet unassigned successors. Then, each available job is processed at a rate proportional to the total weight of its assigned jobs.

\begin{example} \label{exa:single machine non-clairvoyant}
 Consider four jobs with the following weights and processing times
 \begin{center}
  \begin{tabular}{r|rrrr}
   $j$ & $1$ & $2$ & $3$ & $4$ \\\hline
   $w_j$ & $1$ & $2$ & $1$ & $1$ \\
   $p_j$ & $6$ & $4$ & $3$ & $5$
  \end{tabular}
 \end{center}
 and assume that there is a single precedence constraint from job~$1$ to job~$2$. The schedule for this instance resulting from \cref{alg1} is depicted in \cref{fig:single machine non-clairvoyant}, and the resulting completion times are $C_1 = 10$, $C_2 = 17$, $C_3 = 14$, and $C_4 = 18$.
 \begin{figure}[h]
  \centering
  \begin{tikzpicture}[xscale=.7, yscale=1.2]
   \fill[fill=firstcolor, job] (0,.4) rectangle node {$1$} (10,1);
   \fill[fill=thirdcolor, job] (0,.2) rectangle +(10,.2);
   \fill[fill=fourthcolor, job] (0,0) rectangle +(10,.2);
   \fill[fill=secondcolor, job] (10,.5) rectangle node {$2$} +(4,.5);
   \fill[fill=thirdcolor, job] (10,.25) rectangle node {$3$} +(4,.25);
   \fill[fill=fourthcolor, job] (10,0) rectangle node {$4$} +(4,.25);
   \fill[fill=secondcolor, job] (14,1/3) rectangle node {$2$} (17,1);
   \fill[fill=fourthcolor, job] (14,0) rectangle node {$4$} +(3,1/3);
   \fill[fill=fourthcolor, job] (17,0) rectangle node {$4$} +(1,1);
  \end{tikzpicture}
  \caption{Schedule for the instance in \cref{exa:single machine non-clairvoyant}}
  \label{fig:single machine non-clairvoyant}
 \end{figure}
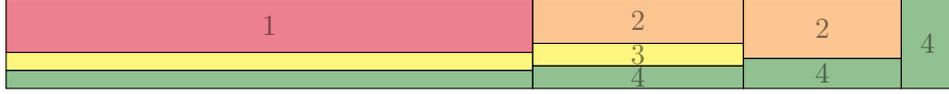
\end{example}

\begin{theorem} \label{thm:non-clrv-single}
 Using \cref{alg1} at time~$0$ and at the first $n-1$ job completion times to determine subsequent processing rates is a $2$-competitive strategy for total weighted completion time minimization.
\end{theorem}
\begin{proof}
 We prove the claim by induction on the number~$n$ of jobs. For a single job, the algorithm obviously computes the optimal schedule. Now consider an instance~$I$ with $n > 1$ jobs, and assume w.l.o.g.\ that $w(N) = 1$ and that job~$1$ is the first job completed in the schedule computed by the algorithm. We denote the the schedule resulting from \cref{alg1} applied to $I$ by $\alg(I)$, and we write $Y_j(t) \coloneqq Y_j^{\alg(I)}(t)$ for $t \ge 0$. Within $[0,C_1^{\alg(I)})$ every job~$i \in F_0$ is processed at a constant rate equal to the total weight of jobs in the tree~$T(i)$. Therefore, $C_1^{\alg(I)} = p_1 / w(T(1))$, and $Y_i(C_1^{\alg(I)}) = w(T(i)) C_1^{\alg(I)}$ for all $i \in F_0$. We define the instance~$I'$ that consists of the jobs $j=2,\dotsc,n$ with processing times $p_j' \coloneqq p_j - Y_j(C_1^{\alg(I)})$. \Cref{alg1} applied to $I$ behaves after time~$C_1^{\alg(I)}$ exactly as if it were applied to $I'$. Hence, $C_j^{\alg(I)} = C_1^{\alg(I)} + C_j^{\alg(I')}$ for all $j \in \{2,\dotsc,n\}$. Therefore,
 \[\sum_{j=1}^n w_j C_j^{\alg(I)} = \sum_{j=1}^n w_j C_1^{\alg(I)} + \sum_{j=2}^n w_j C_j^{\alg(I')} = C_1^{\alg(I)} + \sum_{j=2}^n w_j C_j^{\alg(I')}.\]
 Now consider a fixed non-preemptive optimal schedule~$\opt(I)$ for $I$. By shortening jobs in $\opt(I)$, removing job~$1$, and contracting any occurring idle times, we can transform this to a feasible schedule~$\mathrm S'$ for $I'$.
 Since we contract idle times, the starting time (and, thus, also completion time) of each job~$j$ is lowered by the total time by which all preceding jobs have been shortened.
 By construction of $I'$, only the processing time of jobs $k \in F_0$ is shortened by exactly $Y_k(C_1^{\alg(I)})$. Hence, $C_j^{\opt(I)} - C_j^{\mathrm S'} = \sum_{k \in F_0: C_k^{\opt(I)} \le C_j^{\opt(I)}} Y_k(C_1^{\alg(I)})$ for $j = 2, \dotsc, n$.
 Overall, we obtain
 \begin{align*}
  \sum_{j=1}^n w_j C_j^{\opt(I)} &= w_1 C_1^{\opt(I)} + \sum_{j=2}^n w_j (C_j^{\mathrm S'} + C_j^{\opt(I)} - C_j^{\mathrm S'}) \\
  &\ge \sum_{j=2}^n w_j C_j^{\mathrm S'} + \sum_{j=1}^n w_j \sum_{\substack{k \in F_0 \\ C_k^{\opt(I)} \le C_j^{\opt(I)}}} Y_k(C_1^{\alg(I)}) \\
  &= \sum_{j=2}^n w_j C_j^{\mathrm S'} + \sum_{j=1}^n w_j \sum_{\substack{k \in F_0 \\ C_k^{\opt(I)} \le C_j^{\opt(I)}}} w\bigl(T(k)\bigr) \, C_1^{\alg(I)} \\
  &= \sum_{j=2}^n w_j C_j^{\mathrm S'} + C_1^{\alg(I)} \sum_{i \in F_0} \sum_{j \in T(i)} w_j \sum_{\substack{k \in F_0 \\ C_k^{\opt(I)} \le C_j^{\opt(I)}}} w\bigl(T(k)\bigr) \\
  &\ge \sum_{j=2}^n w_j C_j^{\mathrm S'} + C_1^{\alg(I)} \sum_{i \in F_0} w\bigl(T(i)\bigr) \sum_{\substack{k \in F_0 \\ C_k^{\opt(I)} \le C_i^{\opt(I)}}} w\bigl(T(k)\bigr) \\
  &\ge \sum_{j=2}^n w_j C_j^{\mathrm S'} + C_1^{\alg(I)} \cdot \frac 1 2 \bigg(\sum_{j=1}^n w_j \bigg)^{\!2} \\
  &= \sum_{j=2}^n w_j C_j^{\mathrm S'} + C_1^{\alg(I)} \cdot \frac 1 2 \ge \sum_{j=2}^n w_j C_j^{\opt(I')} + \frac{C_1^{\alg(I)}}{2},
 \end{align*}
 where the second inequality holds because $C_j^{\opt(I)} \ge C_i^{\opt(I)}$ for all $i \in F_0$ and $j \in T(i)$. Thus, we have by induction hypothesis:
 \begin{align*}
  \sum_{j=1}^n w_j C_j^{\alg(I)} &= C_1^{\alg(I)} + \sum_{j=2}^n w_j C_j^{\alg(I')} \\
  &\le 2 \biggl(\frac{C_1^{\alg(I)}}{2} + \sum_{j=2}^n w_j C_j^{\opt(I')} \biggr) \le 2 \sum_{j=1}^n w_j C_j^{\opt(I)}. \qedhere
 \end{align*}
\end{proof}

\section{Non-clairvoyant Scheduling on Identical Parallel Machines} \label{non-clairvoyant parallel}

The difficulty on multiple identical machines is that the total processing power of $m$ cannot be arbitrarily divided among jobs because no job can be processed by more than one machine at the same time. If we imagine that unfinished unavailable jobs pass their weight to their available predecessors, which they can use to ``buy'' processor rate, then on a single machine, they can pass their weight to an arbitrary predecessor because the donated weight will in any case increase the processing rate of this predecessor. This is implemented in \cref{alg1} by simply passing the weight to the available predecessor considered first. Another simple case is when the precedence graph is an out-forest. In this case, every unfinished unavailable job passes its weight to its unique available predecessor, and the rate assignment to the available jobs follows the so-called WDEQ algorithm applied to the collected weights, i.e., as long as there is a job that should receive a rate larger than~$1$, this job receives rate~$1$ and the algorithm recurses on the remaining jobs and machines. This amounts exactly to Algorithm~5 from \textcite{LLMS23}.

In the case that we have multiple machines and arbitrary precedence constraints, the situation is more complicated. For a given unavailable job~$j$, some predecessors may have already collected enough weight to receive an entire processor, while others would still benefit from receiving more weight. So $j$ better passes its weight to the latter in order to reduce its waiting time. It may also be beneficial to split the weight among multiple predecessors. Inspired by an algorithm by \textcite{GGKS19} for the more general setting with online release dates, we model this as a parametric maximum flow problem. For this we adopt a slightly different view than above. We assume that all unfinished jobs can buy ``virtual'' processing rate. If a job is unavailable, the bought virtual rate has to be sent to it via an available predecessor, and the total rate sent through each available job is at most~$1$. This corresponds to a resource allocation problem in a capacitated network. If $\tilde R_j$ is the virtual rate sent to node~$j$, the goal is to balance the ratios $\tilde R_j/w_j$. Since at the end of the day, we are not interested in the exact virtual rates of the unavailable jobs, but only in the total amount sent through each available job, it suffices to minimize the maximum ratio~$\tilde R_j/w_j$ over all unfinished jobs~$j$. Exactly this problem was studied by \citeauthor{GGT89}~\cite[Section~4.1]{GGT89}, who show that this Minimax flow sharing problem can be solved by computing a parametric maximum flow in an extended network. In the following we describe a variant of their construction for our specific case.

Let $t \ge 0$. We define a directed graph~$D_t = (\mathcal V_t, \mathcal A_t)$ as follows: The nodes are $\mathcal V_t \coloneqq U_t \cup \{\mathrm A, \mathrm B, \mathrm Z\}$, and the set of arcs is
\[\mathcal A_t \coloneqq \bigl\{(j,k) \in A \bigm| j,k \in U_t\bigr\} \cup \bigl\{(\mathrm A,\mathrm B)\bigr\} \cup \bigl\{(\mathrm B,j) \bigm| j \in F_t\bigr\} \cup \bigl\{(j,\mathrm Z) \bigm| j \in U_t\bigr\}.\]
We define arc capacities $u_a^t(\pi)$, $a \in \mathcal A_t$, depending on a parameter $\pi > 0$, as follows: We set $u^t_{(\mathrm A, \mathrm B)}(\pi) \coloneqq m$, $u^t_{(\mathrm B, j)}(\pi) \coloneqq 1$ for $j \in F_t$, $u^t_{(j,k)}(\pi) \coloneqq \infty$ for $j,k \in U_t$ with $(j,k) \in A$, and\linebreak $u^t_{(j,\mathrm Z)}(\pi) \coloneqq w_j/\pi$ for $j \in U_t$. This notation allows us to formulate the rate distribution in \cref{alg2}.
\begin{algorithm}
 \caption{Processing rates at time~$t$ on identical parallel machines}
 \label{alg2}
 \begin{algorithmic}
  \If{$\lvert F_t \rvert \leq m$,}
   \State set $R_j(t) \gets 1$ for all $j \in F_t$;
  \Else
   \State compute $\pi_t \gets \max\bigl\{\pi > 0 \bigm| (\{\mathrm A\}, \mathcal V_t \setminus \{\mathrm A\}) \text{ is a minimum-capacity }\mathrm A\text-\mathrm Z\text{-cut w.r.t.\ } u^t(\pi)\bigr\}$,
   \Statex\hspace{\algorithmicindent}and let $x^t$ be a maximum $\mathrm A$-$\mathrm Z$-flow for arc capacities~$u^{t} \coloneqq u^t(\pi_t)$;
   \State set $R_j(t) \gets x^t_{(\mathrm B,j)}$ for all $j \in F_t$. %
  \EndIf
 \end{algorithmic}
\end{algorithm}

\begin{example} \label{exa:non-clairvoyant}
 Assume that there are $m=3$ machines and the following unfinished jobs at time~$t = 0$:
 \begin{center}
  \begin{tabular}{r|*6c}
   $j$ & $1$ & $2$ & $3$ & $4$ & $5$ & $6$ \\\hline
   $p_j$ & $9$ & $9$ & $12$ & $12$ & $9$ & $3$ \\
   $w_j$ & $1$ & $1$ & $1$ & $6$ & $5$ & $1$
  \end{tabular}
 \end{center}
 Assume further that there are the precedence constraints~$\mathcal A_t = \{(1,5),\allowbreak (2,5),\allowbreak (2,6),\allowbreak (3,6),\allowbreak (4,6)\}$. Then the available jobs are $F_t = \{1,2,3,4\}$, so that $\lvert F_t \rvert > m$. \Cref{alg2} considers the directed graph shown in \cref{fig:example non-clairvoyant},
 \begin{figure}[tb]
  \centering
  \begin{tikzpicture}[scale=1.75]
   \draw[Green,very thick] (-2.5,1.5) -- (-2.5,-1.5);
   \draw[Green,very thick] (-1,-1.5) .. controls (1,-0.5) and (3.5,-2) .. (3.75,1.5);
   \node[circle,inner sep=0,minimum size=12pt,draw] (A) at (-3,0) {$\mathrm A$};
   \node[circle,inner sep=0,minimum size=12pt,draw] (B) at (-2,0) {$\mathrm B$};
   \node[circle,inner sep=0,minimum size=12pt,draw] (Z) at (4,0) {$\mathrm Z$};
   \node[rectangle,white,fill=availablecolor,label={[availablecolor]$1$},label={below:$1$}] (1) at (0,1.5) {$1$};
   \node[rectangle,white,fill=availablecolor,label={[availablecolor]$\nicefrac 2 3$},label={below:$1$}] (2) at (0,0.5) {$2$};
   \node[rectangle,white,fill=availablecolor,label={[availablecolor]$\nicefrac 1 3$},label={below:$1$}] (3) at (0,-0.5) {$3$};
   \node[rectangle,white,fill=availablecolor,label={below:$6$}] (4) at (0,-1.5) {$4$};
   \node[availablecolor,fill=white,fill opacity=0.7,above] at (4.north) {$1$};
   \node[rectangle,white,fill=Purple,label={below:$5$}] (5) at (2,1.5) {$5$};
   \node[rectangle,white,fill=Purple,label={$1$}] (6) at (2,-0.5) {$6$};
   \draw[thick,->] (A) -- node[above,fill=white,fill opacity=0.7] {$\textcolor{Blue}{3} \mid \textcolor{Red}{3}$} (B);
   \draw[thick,->] (B) -- node[sloped,above] {$\textcolor{Blue}{1} \mid \textcolor{Red}{1}$} (1);
   \draw[thick,->] (B) -- node[sloped,above,pos=0.66] {$\textcolor{Blue}{\nicefrac{2}{3}} \mid \textcolor{Red}{1}$} (2);
   \draw[thick,->] (B) -- node[sloped,below, pos=0.66] {$\textcolor{Blue}{\nicefrac 1 3} \mid \textcolor{Red}{1}$} (3);
   \draw[thick,->] (B) -- node[sloped,below] {$\textcolor{Blue}{1} \mid \textcolor{Red}{1}$} (4);
   \draw[thick,->] (1) -- node[Blue,sloped,below] {$\nicefrac{7}{9}$} (5);
   \draw[thick,->] (2) -- node[Blue,sloped,below] {$\nicefrac{1}{3}$} (5);
   \draw[thick,->] (2) -- node[Blue,sloped,above] {$\nicefrac{1}{9}$} (6);
   \draw[thick,->] (3) -- node[Blue,sloped,above] {$\nicefrac{1}{9}$} (6);
   \draw[thick,->] (1) to[bend left=45,looseness=1.2] node[above,sloped] {$\textcolor{Blue}{\nicefrac 2 9} \mid \textcolor{Red}{\nicefrac 2 9}$} (Z);
   \draw[thick,->] (2) -- node[above,sloped] {$\textcolor{Blue}{\nicefrac 2 9} \mid \textcolor{Red}{\nicefrac 2 9}$} (Z);
   \draw[thick,->] (3) to[bend right] node[below,sloped,pos=0.6] {$\textcolor{Blue}{\nicefrac 2 9} \mid \textcolor{Red}{\nicefrac 2 9}$} (Z);
   \draw[thick,->] (5) -- node[below,sloped] {$\textcolor{Blue}{\nicefrac{10}{9}} \mid \textcolor{Red}{\nicefrac{10}{9}}$} (Z);
   \draw[thick,->] (6) -- node[above,sloped,pos=.4] {$\textcolor{Blue}{\nicefrac 2 9} \mid \textcolor{Red}{\nicefrac 2 9}$} (Z);
   \draw[thick,->] (4) to[bend right] node[below,sloped] {$\textcolor{Blue}{1} \mid \textcolor{Red}{\nicefrac 4 3}$} (Z);
   \draw[thick,->] (4) -- node[sloped,below,Blue,fill=white,fill opacity=0.7] {$0$} (6);
  \end{tikzpicture}
  \caption{Network~$D_t$ used by \cref{alg2} for the instance from \cref{exa:non-clairvoyant} with \textcolor{Red}{capacities}, \textcolor{Blue}{flows}, \textcolor{Green}{minimum capacity cuts}, and \textcolor{availablecolor}{processing rates} resulting from the parametric maximum flow computation.}\label{fig:example non-clairvoyant}
 \end{figure}
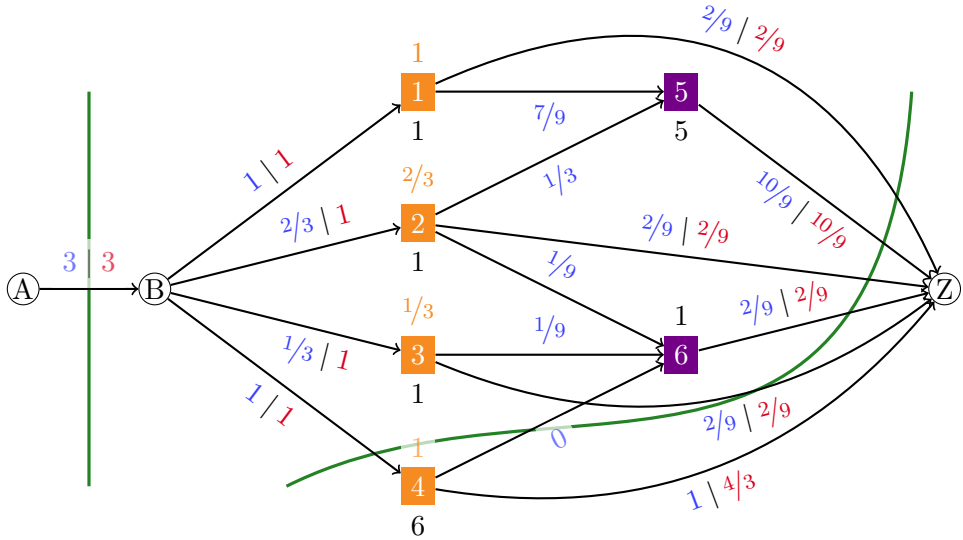
 where the available jobs are colored \availablecolorname, and the unavailable jobs are colored \purplename. The black labels next to the nodes indicate the job weights. Using a parametric flow computation, the algorithm computes $\pi_t = \frac 9 2$, resulting in the arc capacities given in \cref{fig:example non-clairvoyant} in red, as well as the maximum flow~$x^t$ indicated in blue. The rates of the available jobs are set to the incoming flow values and are represented in \availablecolorname{} in the \lcnamecref{fig:example non-clairvoyant}. The two minimum-capacity cuts are shown in green. Note that some of the flow values and rates are not unique. For example, it would also be possible for more flow from node~$\mathrm B$ to node~$5$ to take the route via node~$2$ instead of $1$, resulting in a shift of processing rate from job~$1$ to job~$2$.

 \begin{figure}[h]
  \centering
  \begin{tikzpicture}[xscale=.7, yscale=1.2/2]
  \fill[fill=firstcolor, job] (0,2) rectangle node {$1$} (9,3);
  \fill[fill=secondcolor, job] (0,4/3) rectangle (9,2);
  \fill[fill=thirdcolor, job] (0,1) rectangle (9,4/3);
  \fill[fill=fourthcolor, job] (0,0) rectangle node {$4$} (9,1);
  \fill[fill=secondcolor, job] (9,2) rectangle node {$2$} (12,3);
  \fill[fill=thirdcolor, job] (9,1) rectangle node {$3$} (12,2);
  \fill[fill=fourthcolor, job] (9,0) rectangle node {$4$} (12,1);
  \fill[fill=thirdcolor, job] (12,2) rectangle node {$3$} (18,3);
  \fill[fill=fifthcolor, job] (12,1) rectangle node {$5$} (18,2);
  \fill[fill=fifthcolor, job] (18,2) rectangle node {$5$} (21,3);
  \fill[fill=sixthcolor, job] (18,1) rectangle node {$6$} (21,2);
  \end{tikzpicture}
  \caption{Schedule for the instance in \cref{exa:non-clairvoyant}.}
  \label{fig:virtual schedule non-clairvoyant}
 \end{figure}
 When jobs are scheduled with the rates according to \cref{fig:example non-clairvoyant}, job~$1$ is completed first at time $C_1 = 9$. In all subsequent iterations we always have $|F_t| \leq m$. Therefore, after the completion of job~$1$, all jobs are processed with rate~$1$ leading to the schedule depicted in \cref{fig:virtual schedule non-clairvoyant} with the completion times $C_2 = C_4 = 12$, $C_3 = 18$ and $C_5 = C_6 = 21$.
\end{example}

The parametric minimum cut computation in \cref{alg2} can be carried out with the strongly polynomial algorithm by
\Textcite{GGT89}. This algorithm can be applied to instances with capacities that depend linearly on a parameter~$\lambda$ and computes the minimum cut capacity function. Setting $\lambda \coloneqq 1 / \pi$, the Minimax solution\footnote{Note that in~\cite{GGT89} the algorithms solving the Minimax and the similar Maximin sharing problems are interchanged.} can be obtained as the largest breakpoint~$\lambda_{\max}$ of this function.

The resulting $\pi_t = 1 / \lambda_{\max}$ can be interpreted as the price of processing rate in the following market: A total amount of $m$~units of a single divisible good are sold to $|U_t|$~buyers $j \in U_t$ with budgets~$w_j$. The paths from $\mathrm{B}$ to the nodes $j \in U_t$ represent different possible routes on which the good can be delivered from the supplier to the respective buyers, where the amount of good sent on any arc $(\mathrm{B}, i)$, $i \in F_t$, must not exceed~$1$. The price~$\pi_t$ is the largest possible price for the good so that all units of the good will be sold, and any corresponding maximum flow corresponds to a possible allocation to the buyers. We can alternatively interpret the price of the good as the price for transporting one unit along the arc~$(\mathrm A, \mathrm B)$. \Textcite{JV10} considered the problem where prices for all arcs of a digraph have to be determined, which they solve by similar flow-based methods.

\begin{theorem} \label{thm:non-clrv-P}
 Using \cref{alg2} at time~$0$ and at the first $n-1$ job completion times to determine subsequent processing rates is a $3$-competitive strategy for total weighted completion time minimization.\end{theorem}

In order to prove this theorem consider an arbitrary instance~$I$ and denote by $\alg(I)$ the schedule output by \cref{alg2} for $I$. We call a job~$j \in U_t$ \emph{active} at time~$t$ if one of its predecessors or $j$ itself is processed at a rate smaller than $1$, otherwise it is \emph{inactive}.
The set of active jobs at time~$t$ is denoted by $A_t$. Let $\mu_j^I$ be the total active time of job~$j$ in $\alg(I)$, and let $\lambda_j^I$ be the inactive time before its completion, i.e., $C_j^{\alg(I)} = \mu_j^I + \lambda_j^I$. Clearly, $\lambda_j^I$ is bounded by the maximum total processing time of a precedence chain ending in $j$, which is a lower bound on $C_j^{\opt(I)}$.

Let $I_1$ be the instance with the same job set but with a single machine that is $m$ times faster. We interpret this as still admitting a total processing rate of $m$, which, however, can now be divided arbitrarily among jobs. Then
\begin{equation}
 \sum_{j=1}^n w_j C_j^{\opt(I_1)} \le \sum_{j=1}^n w_j C_j^{\opt(I)} \label{ineq:fast_single}
\end{equation}
because the we relaxed the restriction that no job is processed at a rate larger than $1$. We now bound the total weighted active time.

\begin{lemma} \label{lem:active-time}
$\sum_{j=1}^n w_j \mu_j^I \le 2 \sum_{j=1}^n w_j C_j^{\opt(I_1)}$.
\end{lemma}

\begin{proof}
We prove the \lcnamecref{lem:active-time} by induction on the number of jobs. For a single job, $\mu_1^I = 0$, and the claim is clear. So consider an instance~$I$ with $n > 1$ jobs, and let w.l.o.g.\ job~$1$ be completed first in $\alg(I)$. As in \cref{non-clairvoyant single}, we write $Y_j(t) \coloneqq Y_j^{\alg(I)}(t)$ and $R_j(t) \coloneqq R_j^{\alg(I)}(t)$ for $t \ge 0$. We also consider the instances~$I'$ and $I_1'$ with jobs $2,\dotsc,n$ and processing times $p_j' \coloneqq p_j - Y_j(C_1^{\alg(I)})$ on $m$ parallel machines and on a single fast machine, respectively. Then for every $j \in \{2,\dotsc,n\}$ we have
\[\mu_j^I = \begin{cases*} C_1^{\alg(I)} + \mu_j^{I'} &if $j \in A_0$;\\ \mu_j^{I'} &else.\end{cases*}\]
Therefore,
\[
 \sum_{j=1}^n w_j \mu_j^I = \sum_{j \in A_0} w_j C_1^{\alg(I)} + \sum_{j=2}^n w_j \mu_j^{I'} \le w(A_0) C_1^{\alg(I)} + 2 \sum_{j=2}^n w_j C_j^{\opt(I'_1)},
\]
where the last inequality holds by induction. It remains to show that the right-hand side is bounded by $2 \sum_{j=1}^n w_j C_j^{\opt(I_1)}$.

If $\lvert F_0 \rvert \leq m$, then $A_0 = \emptyset$ and the claim is satisfied because increasing the processing times and adding a job cannot improve the optimal objective value. So assume from now on that ${\lvert F_0 \rvert > m}$. Then the algorithm computes the value~$\pi_0$ as well as flows and capacities $x^0 \le u^0$. Since $(\{\mathrm A\},\allowbreak {\mathcal V_t \setminus \{\mathrm A\}})$ is a minimum-capacity cut, it is fully saturated by the maximum flow $x^0$, meaning that $x^0$ has flow value~$m$. This implies that
\begin{equation}
 \sum_{i \in F_0} R_i(0) = \sum_{i \in F_0} x_{(\mathrm B, i)}^0 = m, \label{eq:full_capacity}
\end{equation}
i.e., the algorithm utilizes the total available processor capacity. Apart from the cut $(\{\mathrm A\}, \mathcal V_t \setminus \{\mathrm A\})$, there is another minimum-capacity $\mathrm A$-$\mathrm Z$-cut~$(\mathcal S, \mathcal V_t \setminus \mathcal S)$, which is crossed only by arcs from $\{\mathrm B\} \times F_0$ or $U_0 \times \{\mathrm Z\}$ because $\pi_0$ is a breakpoint of the minimum cut capacity function.
For every $j \in N$ we have
\begin{equation}
 x_{(j,\mathrm Z)}^0 \le u_{(j,\mathrm Z)}^0 = \frac{w_j}{\pi_0}. \label{ineq:flow_weight}
\end{equation}
This inequality is tight for all active jobs~$j \in A_0$ because either $j$ is available and processed with $R_j(0) < 1$ or there is an available predecessor~$i \in F_0$ with $R_i(0) < 1$. Since $x_{(\mathrm B, i)}^0 = R_i(0) < 1 = u_{(\mathrm B, i)}^0$, the arc $(\mathrm B, i)$ does not cross the cut $(\mathcal S, \mathcal V_t \setminus \mathcal S)$. Therefore, $j \in \mathcal S$, whence $(j,\mathrm Z)$ crosses the cut, implying that $x_{(j,\mathrm Z)}^0 = u^0_{(j,\mathrm Z)}$.

Let $\tilde{w}_i \coloneqq \pi_0 R_i(0)
$ for all $i \in F_0$. Then
\begin{equation}
 Y_k(C_1^{\alg(I)}) = R_k(0) \, C_1^{\alg(I)} = \frac{\tilde w_k}{\pi_0} \, C_1^{\alg(I)} \label{eq:elapsed_time}
\end{equation}
for all $k \in F_0$. Moreover,
\begin{equation} w(A_0) = \sum_{j \in A_0} \pi_0 \, x_{(j,\mathrm Z)}^0 \le \sum_{i \in F_0} \pi_0 \, x_{(\mathrm B,i)}^0 = \sum_{i \in F_0} \pi_0 \, R_i(0) = \tilde w(F_0)\label{ineq:weighted_delay_ALG}\end{equation}
because all flow to some node~$j \in A_0$ has to pass an arc from $\{\mathrm B\} \times F_0$.

Now consider a fixed optimal schedule $\opt(I_1)$ for $I_1$. By shortening the jobs, removing job~$1$, and contracting occurring idle times, we obtain a feasible schedule~$\mathrm S'$ for $I_1'$. Then
\begin{align}
 \sum_{j=1}^n w_j C_j^{\opt(I_1)} &= w_1 C_1^{\opt(I_1)} + \sum_{j=2}^n w_j (C_j^{\mathrm S'} + C_j^{\opt(I_1)} - C_j^{\mathrm S'}) \notag \\
 &\ge \sum_{j=2}^n w_j C_j^{\mathrm S'} + \sum_{j=1}^n w_j \sum_{\substack{k \in F_0\\ C_k^{\opt(I_1)} \le C_j^{\opt(I_1)}}} \frac{Y_k(C_1^{\alg(I)})}{m} \notag \\
 &\stackrel{\eqref{ineq:flow_weight}}\ge \sum_{j=2}^n w_j C_j^{\mathrm S'} + \pi_0 \sum_{j=1}^n x_{(j,\mathrm Z)}^0 \sum_{\substack{k \in F_0\\ C_k^{\opt(I_1)} \le C_j^{\opt(I_1)}}} \frac{Y_k(C_1^{\alg(I)})}{m}. \label{umformung}
\end{align}
The flow~$x^0$ can be decomposed into path flows $x_P^0$ for $\mathrm A$-$\mathrm Z$-paths $P$ in $D_0$ such that $\sum_{P \ni a} x_{P}^0 = x_a^0$ for all $a \in \mathcal A_0$.
Using the path decomposition, we can express the flow on the arcs $(\mathrm B, i)$, $i \in F_0$, and $(j, \mathrm Z)$, $j \in U_0$, as
\begin{align*}
  x_{(\mathrm B, i)}^0 &= \sum_{P \ni (\mathrm B, i)} x_P^0 = \sum_{j=1}^n \sum_{P \ni  (\mathrm B,i), (j,\mathrm Z)} x_P^0
& \text{and} &&
  x_{(j,\mathrm Z)}^0 &= \sum_{P \ni (j,\mathrm Z)} x_P^0 = \sum_{i \in F_0}\sum_{P \ni (\mathrm B,i), (j,\mathrm Z)} x_P^0.
\end{align*}
We use this in order to bound the second sum in~\eqref{umformung} and obtain
\begin{align*}
 \sum_{j=1}^n  w_j C_j^{\opt(I_1)} &\geq
 \sum_{j=2}^n w_j C_j^{\mathrm S'} + \pi_0
 \sum_{j=1}^n \sum_{i \in F_0}\sum_{P \ni (\mathrm B,i), (j,\mathrm Z)} x_P^0 \sum_{\substack{k \in F_0\\ C_k^{\opt(I_1)} \le C_j^{\opt(I_1)}}} \frac{Y_k(C_1^{\alg(I)})}{m} \\
 &\ge
 \sum_{j=2}^n w_j C_j^{\mathrm S'} + \pi_0
 \sum_{i \in F_0} \sum_{j=1}^n \sum_{P \ni  (\mathrm B,i), (j,\mathrm Z)} x_P^0 \sum_{\substack{k \in F_0\\ C_k^{\opt(I_1)} \le C_i^{\opt(I_1)}}} \frac{Y_k(C_1^{\alg(I)})}{m}.
\end{align*}
In the inequality we used that every path $P$ in $D_0$ containing $i \in F_0$ and $j \in N$ corresponds to a precedence chain from job~$i$ to job~$j$. Therefore, if such a path exists, $C_i^{\opt(I_1)} \le C_j^{\opt(I_1)}$, and thus, the set of jobs~$k \in F_0$ with $C_k^{\opt(I_1)} \le C_i^{\opt(I_1)}$ is included in the set of jobs~$k \in F_0$ completed before time~$C_j^{\opt(I_1)}$.
  Finally, we compute
\begin{align*}
 \sum_{j=1}^n w_j C_j^{\opt(I_1)} &\ge \sum_{j=2}^n w_j C_j^{\mathrm S'} + \pi_0 \sum_{i \in F_0} x_{(\mathrm B, i)}^0 \sum_{\substack{k \in F_0\\ C_k^{\opt(I_1)} \le C_i^{\opt(I_1)}}} \frac{Y_k(C_1^{\alg(I)})}{m} \\
 &\stackrel{\mathclap{\eqref{eq:elapsed_time}}}= \sum_{j=2}^n w_j C_j^{\mathrm S'} + \sum_{i \in F_0} \tilde w_i \sum_{\substack{k \in F_0\\ C_k^{\opt(I_1)} \le C_i^{\opt(I_1)}}} \tilde w_k \frac{C_1^{\alg(I)}}{\pi_0 m} \\
 &\ge \sum_{j=2}^n w_j C_j^{\mathrm S'} + \frac 1 2 \, \tilde w(F_0)^2 \, \frac{C_1^{\alg(I)}}{\pi_0 \, m} \\
 &\stackrel{\mathclap{\eqref{eq:full_capacity}}}= \sum_{j=2}^n w_j C_j^{\mathrm S'} + \frac 1 2 \, \tilde w(F_0) \, C_1^{\alg(I)} \\
 &\stackrel{\mathclap{\eqref{ineq:weighted_delay_ALG}}}\ge \sum_{j=2}^n w_j C_j^{\opt(I_1')} + \frac 1 2 \, w(A_0) \, C_1^{\alg(I)}. \qedhere
\end{align*}
\end{proof}

Using the \lcnamecref{lem:active-time}, we obtain
\begin{align*}
 \sum_{j=1}^n w_j C_j^{\alg(I)} &= \sum_{j=1}^n w_j \lambda_j + \sum_{j=1}^n w_j \mu_j \\
 &\le \sum_{j=1}^n w_j C_j^{\opt(I)} + 2 \sum_{j=1}^n w_j C_j^{\opt(I_1)} \stackrel{\eqref{ineq:fast_single}}\le 3 \sum_{j=1}^n w_j C_j^{\opt(I)},
\end{align*}
concluding the proof of \cref{thm:non-clrv-P}.

\section{Clairvoyant Approximation Algorithms} \label{approximations}

When processing times are known in advance, the non-clairvoyant algorithms described above can be simulated in order to compute a \emph{virtual} preemptive schedule. On a single machine this virtual schedule can be transformed to a schedule without preemptions, while on identical parallel machines the number of preemptions can be reduced to $O(n^2)$. Both transformations do not increase the completion time of any job.

\paragraph{Non-preemptive scheduling on a single machine}

On a single machine, after computing the virtual preemptive schedule, we can simply perform list scheduling in the order of the virtual completion times. Since the preemptive schedule is not actually executed, we also have to compute its completion times (instead of observing them in the non-clairvoyant setting). This is specified in \cref{alg3}.
\begin{algorithm}
 \caption{Clairvoyant scheduling with precedence constraints on a single machine}
 \label{alg3}
 \begin{algorithmic}
  \State Initialize $t \gets 0$, $U \gets N$, $W \gets \sum_{j \in N} w_j$, and $Y_j \gets 0$ for all $j \in N$.
  \State Perform depth-first search and store total weights~$W(j)$ of subtrees rooted at each \Comment{$O(n^2)$}
  \Statex\hspace{\algorithmicindent}node~$j \in N$.
  \While{$U \neq \emptyset$,} \Comment{$n$ iterations}
   \State let $F$ be the jobs from $U$ without predecessor in $U$;
   \ForAll{$i \in F$} \Comment{$O(n)$ iterations}
    \State let $R_i \gets \frac{W(i)}{W}$;
    \State let $\tau_i \gets \frac{p_i-Y_i}{R_i}$;
   \EndFor
   \State let $j \gets \argmin_{i \in F} \tau_i$; \Comment{$O(n)$}
   \State set $Y_i \gets Y_i + R_i \cdot \tau_j$ for all $i \in F$;
   \State update $t \gets t + \tau_j$, $U \gets U \setminus \{j\}$, and $W \gets W - w_j$;
   \State set $C_j' \gets t$.
  \EndWhile
 \State Perform list scheduling in order $C_j'$. \Comment{$O(n)$}
\end{algorithmic}
\end{algorithm}

\begin{example}
  \begin{figure}[h]
  \centering
  \begin{tikzpicture}[xscale=.7, yscale=1.2]
   \fill[fill=firstcolor, job] (0,0) rectangle node {$1$} (6,1);
   \fill[fill=thirdcolor, job] (6,0) rectangle node {$3$} +(3,1);
   \fill[fill=secondcolor, job] (9,0) rectangle node {$2$} +(4,1);
   \fill[fill=fourthcolor, job] (13,0) rectangle node {$4$} +(5,1);
  \end{tikzpicture}
  \caption{Schedule obtained by list scheduling in order of the completion times from the schedule in \cref{exa:single machine non-clairvoyant}}
  \label{fig:schedule non-preemptive}
 \end{figure}
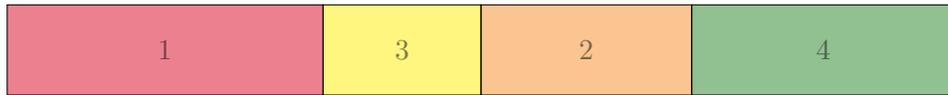
 The schedule resulting from \cref{alg3} for the instance from \cref{exa:single machine non-clairvoyant} is illustrated in \cref{fig:schedule non-preemptive}.
 Its total weighted completion time is $\sum_{j=1}^4 w_j C_j^{\alg} = 6 + 9 + 2 \cdot 13 + 18 = 59$. An optimal schedule processes the jobs in the order $3, 1, 2, 4$ and has objective value $\sum_{j=1}^4 w_j C_j^{\opt} = 3 + 9 + 2 \cdot 13 + 18 = 56$. The example demonstrates that the schedule resulting from our approximation algorithm is not consistent with a Sidney decomposition, in contrast to (almost) all previously known $2$-approximation algorithms~\cite{CS05}.
\end{example}

\begin{corollary}
 \Cref{alg3} is a $2$-approximation algorithm that runs in time $O(n^2)$.
\end{corollary}
\begin{proof}
 In this proof we omit the instance from our notations. We assume w.l.o.g.\ that the jobs are scheduled in the order $1,\dotsc,n$ in the last step of the algorithm, so that $C_1^{\alg} < \cdots < C_n^{\alg}$. Then for every job~$j \in N$ we have $C_j^\alg = \sum_{k \le j} p_k \le C_j'$ because all jobs~$k$ with $k \le j$ have been processed to completion in the virtual schedule by time~$C_j'$. Hence, \[\sum_{j=1}^n w_j C_j^{\alg} \le \sum_{j=1}^n w_j C_j' \stackrel{\ref{thm:non-clrv-single}}\le 2 \sum_{j=1}^n w_j C_j^{\opt}.\]

 Since in every iteration, one job is removed from $U$, the while-loop has at most $n$ iterations. In each iteration a linear number of elementary operations is performed to compute the rates and to find the job finished next, so that the total number of steps in the loop is in $O(n^2)$. Also list scheduling can be done in linear time, whence the total number of operations is bounded by a polynomial in $O(n^2)$.
\end{proof}

Note that this bound on the number of elementary operations does not immediately imply strongly polynomial running time because it also has to be ensured that the encoding length of the numbers occurring in the computation remains polynomially bounded. This will be shown in \cref{apx:encoding_length}.

\paragraph{Finitely many preemptions on identical parallel machines}

On identical parallel machines, we apply \citeauthor{McN59}'s~\cite{McN59} wrap-around rule to each piece of the virtual schedule between two consecutive completion times. This is formalized in \cref{alg4}.

\begin{algorithm}[t]
 \caption{Clairvoyant scheduling with precedence constraints on identical parallel machines}
 \label{alg4}
 \begin{algorithmic}
  \LComment{Compute virtual schedule}
  \State Initialize $t \gets 0$, $U \gets N$, and $Y_j(t) \gets 0$ for all $j \in N$.
  \While{$U \neq \emptyset$} \Comment{$n$ iterations}
   \State let $F_t$ be the jobs from $U$ without predecessor in $U$;
   \State apply \cref{alg2} to the graph~$D_t$ and the parametric capacities~$u^{t}$ defined  \Comment{$O(n^3)$}
   \Statex\hspace{\algorithmicindent}in \cref{non-clairvoyant parallel} to obtain $R_i(t)$ for all $i \in F_t$;
   \State set $\tau_i \gets \frac{p_i - Y_i(t)}{R_i(t)}$ for all $i \in F_t$; \Comment{$O(n)$}
   \State let $j \gets \argmin_{i \in F_t} \tau_i$; \Comment{$O(n)$}
   \State set $Y_i(t+\tau_j) \gets Y_i(t) + R_t(t) \tau_j$ for all $i \in F_t$ and $Y_k(t+\tau_j) \gets Y_k(t)$ for $k \in N \setminus F_t$; \Comment{$O(n)$}
   \State update $t \gets t + \tau_j$ and $U \gets U \setminus \{j\}$;
   \State set $C_j' \gets t$.
  \EndWhile
  \LComment{Compute actual schedule (\citeauthor{McN59}'s wrap around rule)}
  \State Order the jobs so that $C_1' \le \cdots \le C_n'$.
  \State Reset $t \gets 0$.
  \For{$j=1,\dotsc,n$} \Comment{$n$ iterations}
   \State set $i \gets 1$, $u \gets t$, $k \gets j$, and $\Delta_k \gets Y_k(C_j') - Y_k(t)$;
   \While{$k \le n$} \Comment{$\le 2n$ iterations}
    \If{$u + \Delta_k \le C_j'$}
     \State process job~$k$ on machine~$i$ from time $u$ to time $u+\Delta_k$;
     \State replace $u \gets u + \Delta_k$;
     \State increment $k \gets k + 1$;
     \If{$k \le n$} set $\Delta_k \gets Y_k(C_j') - Y_k(t)$; \EndIf
    \Else
     \State process job~$k$ on machine~$i$ from time $u$ to time $C_j'$;
     \State set $\Delta_k \gets u + \Delta_k - C_j'$;
     \State increment $i \gets i + 1$;
     \State set $u \gets t$;
    \EndIf
   \EndWhile
   \State set $t \gets C_j'$.
  \EndFor
 \end{algorithmic}
\end{algorithm}

\begin{example}
 \begin{figure}[h]
 \centering
 \begin{tikzpicture}[xscale=.7, yscale=1.2/2]
   \node[black, anchor=east] at (0, 2.5) {$i = 1$};
   \node[black, anchor=east] at (0, 1.5) {$i = 2$};
   \node[black, anchor=east] at (0, 0.5) {$i = 3$};
   \fill[fill=firstcolor, job] (0,2) rectangle node {$1$} (9,3);
   \fill[fill=secondcolor, job] (0,1) rectangle node {$2$} (6,2);
   \fill[fill=fourthcolor, job] (6,1) rectangle node {$4$} (9,2);
   \fill[fill=fourthcolor, job] (0,0) rectangle node {$4$} (6,1);
   \fill[fill=thirdcolor, job] (6,0) rectangle node {$3$} (9,1);
   \fill[fill=secondcolor, job] (9,2) rectangle node {$2$} (12,3);
   \fill[fill=fourthcolor, job] (9,1) rectangle node {$4$} (12,2);
   \fill[fill=thirdcolor, job] (9,0) rectangle node {$3$} (12,1);
   \fill[fill=fifthcolor, job] (12,1) rectangle node {$5$} (18,2);
   \fill[fill=thirdcolor, job] (12,2) rectangle node {$3$} (18,3);
   \fill[fill=fifthcolor, job] (18,2) rectangle node {$5$} (21,3);
   \fill[fill=sixthcolor, job] (18,1) rectangle node {$6$} (21,2);
   \end{tikzpicture}
 \caption{Schedule obtained by \cref{alg4} for the instance from \cref{exa:non-clairvoyant}}
 \label{fig:schedule mcnaughton}
\end{figure}
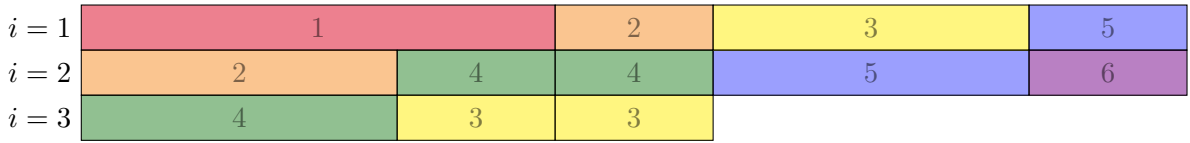
The schedule resulting from \cref{alg4} for the instance from \cref{exa:non-clairvoyant} is illustrated in \cref{fig:schedule mcnaughton}.
The virtual completion times~$C'_j$ computed in the first part are exactly the completion times from \cref{exa:non-clairvoyant}. In the second part of the \lcnamecref{alg4}, \citeauthor{McN59}'s wrap around rule considers the jobs in order of their virtual completion times, i.e., in the order $(1, 2, 4, 3, 5, 6)$.
\end{example}

\begin{corollary}
 \Cref{alg4} is a $3$-approximation ratio introducing at most $O(n^2)$ preemptions. It runs in time $O(n^4)$.
\end{corollary}
\begin{proof}
 The approximation factor follows from \cref{thm:non-clrv-P} because $C_j^{\alg} \le C_j'$ for all $j \in N$. Between any two consecutive virtual completion times, we preempt every unfinished job at most twice: once if it is wrapped around and once at the end of the interval. Hence, the number of preemptions can be bounded by $2n^2$.

 At the beginning and for the first $n-1$ virtual job completions, we apply the algorithm of \textcite{GGT89} to a network with $O(n)$ nodes and $O(n^2)$ arcs. This can be done in time $O(n |A| \log(n^2/|A|)) \subseteq O(n^3)$. This yields a total running time in $O(n^4)$. Clearly this dominates all other steps of the algorithm.
\end{proof}

In \cref{apx:encoding_length} a bound on the encoding length of the appearing numbers is proved, implying that the algorithm runs in strongly polynomial time.

\section{Conclusion}

For single machine scheduling with precedence constraints, under UGC no approximation ratio better than $2$ can be achieved within polynomial time. On the other hand, for non-clairvoyant scheduling, even without precedence constraints, no performance guarantee below $2$ is possible, no matter how much computation time is allowed. Our algorithm shows that, when allowing infinitesimal preemptions, these two bounds can be reached simultaneously, i.e., there is a $2$-competitive efficient non-clairvoyant algorithm. In other words, assuming UGC, the problem is so hard that knowing the processing times is of no use when polynomial running time is required. This is in contrast to many scheduling problems without precedence constraints, which also face a lower bound of $2$ for any non-clairvoyant algorithm but can be solved in polynomial time or admit a PTAS in the clairvoyant setting~\cite{Smi56,ABC+99}.

Also for preemptive scheduling on identical parallel machines, the presented algorithm has the best known performance guarantee of any approximation algorithm and of any non-clairvoyant algorithm, although no matching (conditional) lower bounds exist in this case. So it remains open to determine the exact approximability as well as the best possible competitive ratio of a non-clairvoyant algorithm.

\printbibliography

\appendix

\section{Polynomial Encoding Length} \label{apx:encoding_length}

As mentioned in the introduction, if \cref{alg2} is applied to a single machine instance, it will return the same processing rates as \cref{alg1}. Thus, also the virtual schedules computed in \cref{alg3,alg4} coincide on a single machine. Consequently, it suffices to show that the encoding length of numbers occurring in \cref{alg4} is polynomially bounded.

\begin{lemma}
 The binary encoding length of all numbers occurring in \cref{alg4} bounded by a polynomial in the input encoding length.
\end{lemma}
\begin{proof}
 We assume w.l.o.g.\ that the input is integral and that the jobs are scheduled in the order $1,\dotsc,n$ in the last step of the algorithm, so that $C_1^{\alg} < \cdots < C_n^{\alg}$.

 For every iteration~$j$ in which $j$ is completed in the virtual schedule let $t_j$ be the value of $t$ at the beginning of the iteration, i.e., $t_1 = 0$, $t_2 = C_1'$ etc. If $|F_{t_j}| \le m$, all $R_k(t_j)$, $k \in F_{t_j}$, are set to one, so that only additions and subtractions take place in iteration~$j$. If $|F_{t_j}| > m$, \cref{alg2} simply invokes the smallest breakpoint algorithm by \textcite{GGT89}. Since this runs in strongly polynomial time, the results have polynomial encoding length with respect to the input. The crucial observation here is that the input passed to \cref{alg2} only consists of a subset of jobs with weights and precedence constraints, as well as the number of machines, which are all original input parameters for \cref{alg4}. In particular, it does not include the current clock time~$t_j$ or the elapsed processing times~$Y_k(t_j)$ of jobs $k \in N$. Therefore, all $\pi_{t_j}$ are of the form $\alpha_j/\rho_j$ with $\alpha_j, \rho_j \in \mathbb Z$, whose encoding length is polynomial in the original input. Since this leads to arc capacities from the set $1/\alpha_j \cdot \mathbb Z$, we can assume that also all flow values of $x^{t_j}$ are integer multiples of $1/\alpha_j$, i.e., there are $\upsilon_{jk} \in \mathbb Z$ of polynomial encoding length such that $R_k(t_j) = x_{(\mathrm B, k)}^{t_j} = \upsilon_{jk}/\alpha_j$ for $j \in N$. We show by induction on $j$ that all $t_j$ and $Y_k(t_j)$ are integer multiples of $1/(\prod_{i=1}^{j-1} \upsilon_{ii})$. The claim is true for $j = 1$, as all $Y_k$ and $t$ are zero. Now assume the claim is known for some $j < n$. The next values $t_{j+1}$ and $Y_k(t_{j+1})$, $k \in F_{t_j}$, are computed in iteration~$j$ as
 \begin{align*}
  t_{j+1} &= t_j + \tau_j = t_j + \frac{1}{R_j(t_j)} \cdot (p_j - Y_j(t_j)) = t_j + \frac{\alpha_j}{\upsilon_{jj}} \cdot (p_j - Y_j(t_j)) \in \frac{1}{\prod_{i=1}^{j} \upsilon_{ii}} \cdot \mathbb Z, \\
  Y_k(t_{j+1}) &= Y_k(t_j) + R_k(t_j) \cdot \tau_j = \frac{R_k(t_j)}{R_j(t_j)} \cdot (p_j - Y_j(t_j)) = \frac{\upsilon_{jk}}{\upsilon_{jj}} \cdot (p_j - Y_j(t_j)) \in \frac{1}{\prod_{i=1}^{j} \upsilon_{ii}} \cdot \mathbb Z,
 \end{align*}
 while the $Y_k(t_{j+1})$ with $k \in N \setminus F_{t_j}$ stay unchanged.
 Since $\prod_{i=1}^n \upsilon_{ii} \le (\max_{i=1}^n \upsilon_{ii})^n$, the denominators of all $t_j$ and $Y_k(t_j)$ have encoding length bounded by $n \log (\max_i \upsilon_{ii})$, which is polynomially bounded. Since $t_j$ and $Y_k(t_j)$ always take values between $0$ and $\sum p_j$, their numerators are bounded by $(\sum_{j=1}^n p_j) \cdot \prod_{i=1}^{n} \upsilon_{ii}$. The other numbers~$R_i(t_j)$ and $\tau_i$, occurring in the virtual schedule computation, are set in every iteration to the result of a new arithmetic expression in the input and the values of $t_j$ and $Y_i(t_j)$, so that also their encoding length stays polynomially bounded. Finally, this also holds for the wrap-around part, in which additions and subtractions are performed.
\end{proof}

\end{document}